\documentclass[a4,11pt]{article}

\usepackage{tikz}
\usepackage{calc}
\usepackage{amsfonts}
\usepackage{amsthm}
\usepackage{amssymb}
\usepackage{mathtools}
\usepackage{subcaption}
\usepackage{caption}
\usepackage{url}

\usetikzlibrary{shapes,positioning,chains,fit,calc,decorations.markings,decorations.pathreplacing}
\tikzstyle{ccyan}=[circle, draw, thick,fill=cyan!30, minimum size=12pt,inner sep=0pt]
\tikzstyle{cgrey}=[circle, draw, thick,fill=gray!30, minimum size=10pt,inner sep=0pt]
\tikzstyle{cgreys}=[circle, draw, thick,fill=gray!30, minimum size=12pt,inner sep=0pt]

\newcommand{\ket}[1]{| #1 \rangle}
\newcommand{\bra}[1]{\langle #1 |}
\newcommand{\braket}[2]{\langle #1 | #2 \rangle}

\def\U{\uparrow}
\def\D{\downarrow}
\def\L{\leftarrow}
\def\R{\rightarrow}
\def\S{\circlearrowleft}

\newcommand{\comment}[1]{}

\newtheorem{lemma}{Lemma}

\title{Lackadaisical quantum walks with multiple marked vertices}

\author{Nikolajs Nahimovs}
\date{\small{Center for Quantum Computer Science, Faculty of Computing, University of Latvia} \\ 
\small{Raina bulv. 19, Riga, LV-1586, Latvia}\\
\small{\texttt{nikolajs.nahimovs@lu.lv}}}

\begin{document}

\maketitle


\begin{abstract}

\noindent
The concept of lackadaisical quantum walk -- quantum walk with self loops -- was first introduced for discrete-time quantum walk on one-dimensional line \cite{Norio:2005}. Later it was successfully applied to improve the running time of the spacial search on two-dimensional grid~\cite{Wong:2018}.

In this paper we study search by lackadaisical quantum walk on the two-dimensional grid with multiple marked vertices. First, we show that the lackadaisical quantum walk, similarly to the regular (non-lackadaisical) quantum walk, has exceptional configuration, i.e. placements of marked vertices for which the walk has no speed-up over the classical exhaustive search.
Next, we demonstrate that the weight of the self-loop suggested in \cite{Wong:2018} is not optimal for multiple marked vertices. And, last, we show how to adjust the weight of the self-loop to overcome the aforementioned problem.

\end{abstract}

 
\section{Introduction}

Quantum walks are quantum counterparts of classical random walks \cite{Portugal:2013}. 
Similarly to classical random walks, there are two types of quantum walks: discrete-time quantum walks (DTQW),  introduced by Aharonov~{\it et al.}~\cite{Aharonov:1993}, and continuous-time quantum walks (CTQW), introduced by Farhi~{\it et al.}~\cite{Farhi:1998}.
For the discrete-time version, the step of the quantum walk is usually given by two operators -- coin and shift -- which are applied repeatedly. 
The coin operator acts on the internal state of the walker and rearranges the amplitudes of going to adjacent vertices. The shift operator moves the walker between the adjacent vertices.

Quantum walks have been useful for designing algorithms for a variety of search problems\cite{Nagaj:2011}.
To solve a search problem using quantum walks, we introduce the notion of marked elements (vertices), corresponding to elements of the search space that we want to find.
We perform a quantum walk on the search space with one transition rule at the unmarked vertices, and another transition rule at the marked vertices. If this process is set up properly, it leads to a quantum state in which the marked vertices have higher probability than the unmarked ones. This method of search using quantum walks was first introduced in \cite{Shenvi:2003} and has been used many times since then.

Most of the papers studying quantum walks consider a search space containing a single marked element only. 
However, in contrary of classical random walks, the behavior of the quantum walk can drastically change if the search space contains more that one marked element. 
Ambainis and Rivosh~\cite{Ambainis:2008} have studied DTQW on two-dimensional grid and showed that if the diagonal of the grid is fully marked then the probability of finding a marked element does not grow over time.
Wong~\cite{Wong:2016} analyzed the spatial search problem by CTQW on the simplex of complete graphs and showed that the placement of marked vertices can dramatically influence the required jumping rate of the quantum walk. 
Wong and Ambainis~\cite{Wong:2015} analysed DTQW on the simplex of complete graphs and showed that if one of the complete graphs is fully marked then there is no speed-up over classical exhaustive search.
Nahimovs and Rivosh~\cite{Nahimovs:2015a,Nahimovs:2015} studied DTQW on two-dimensional grid for various placements of multiple marked vertices and proved several gaps in the running time of the walk (depending on the placement of marked vertices). Additionally the authors have demonstrated placements of a constant number of marked vertices for which the walk have no speed-up over classical exhaustive search. They named such placements \textit{exceptional configurations}.
Nahimovs and Santos~\cite{Nahimovs:2017} have extended their work to general graphs.

The concept of lackadaisical quantum walk (quantum walk with self loops) was first studied for DTQW on one-dimensional line \cite{Norio:2005,Stefanak:2014}.
Later on, Wong showed an example of how to apply the self-loops to improve the DTQW based search on the complete graph~\cite{Wong:2015a} and two-dimensional grid~\cite{Wong:2018}.
The running time of the lackadaisical walk heavily depends on a weight of the self-loop. Saha~{\it et al.}\cite{Saha:2018} showed that the weight $l = \frac{4}{N}$ suggested by Wong for two-dimensional grid with a single marked vertex is not optimal for multiple marked vertices. They have demonstrated that for a block of $\sqrt{m} \times \sqrt{m}$ marked vertices one should use the weight $l = \frac{4}{N(m + \sqrt{m}/2)}$.

In this paper, we study search by discrete-time lackadaisical quantum walk on two-dimensional grid with multiple marked vertices.
First, we show that the lackadaisical quantum walk, similarly to the regular (non-lackadaisical) quantum walk, has exceptional configuration, i.e. placements of marked vertices for which the walk have no speed-up over the classical exhaustive search.
Next, we study an arbitrary placement of $m$ marked vertices and demonstrate that the weight $l$ suggested by Wong is not optimal for multiple marked vertices. The same holds for the weight suggested by Saha~{\it et al.}, which seems to work only for a block of $\sqrt{m} \times \sqrt{m}$ marked vertices.
Last, we analyze how to adjust the weight to overcome the aforementioned problem.
We propose two better constructions -- $l = \frac{4m}{N}$ and $l = \frac{4(m - \sqrt{m})}{N}$ -- and discuss their boundaries of application.


\section{Quantum walk on the two-dimensional grid}

\subsection{Regular (non-lackadaisical) quantum walk}

Consider a two-dimensional grid of size $\sqrt{N}\times\sqrt{N}$ with periodic (torus-like) boundary conditions.  
The locations of the grid are labeled by the coordinates $(x,y)$ for $x, y \in \{0,\dots,\sqrt{N}-1\}$.
The coordinates define a set of state vectors, $\ket{x,y}$, which span the Hilbert space ${\cal{H_P}}$ associated with the position. 
Additionally, we define a 4-dimensional Hilbert space ${\cal{H_C}}$, spanned by the set of states $\{\ket{c}: c\in \{\U,\D,\L,\R \}\}$, associated with the direction. We refer to it as the coin subspace. The  Hilbert space of the quantum walk is ${\cal{H_P}}\otimes{\cal{H_C}}$.

The evolution of a state of the walk (without searching) is driven by the unitary operator $U = S\cdot (I\otimes C)$, where $S$ is the flip-flop shift operator
\begin{eqnarray}
S\ket{i,j,\U} & = & \ket{i,j+1,\D} \\
S\ket{i,j,\D} & = & \ket{i,j-1,\U} \\
S\ket{i,j,\L} & = & \ket{i-1,j,\R} \\
S\ket{i,j,\R} & = & \ket{i+1,j,\L},
\end{eqnarray}
and $C$ is the coin operator, given by the Grover's diffusion transformation 
\begin{equation}
C = 2 \ket{s_c}\bra{s_c} - I_4
\end{equation}
with 
$$
\ket{s_c} = \frac{1}{\sqrt{4}}(\ket{\U} + \ket{\D} + \ket{\L} + \ket{\R}) .
$$
The system starts in 
\begin{equation}\label{eq:psi0_grid}
\ket{\psi(0)} = \frac{1}{\sqrt{N}} \sum_{i,j=0}^{\sqrt{N}-1} \ket{i,j} \otimes \ket{s_c} ,
\end{equation}
which is uniform distribution over vertices and directions. Note, that this is a unique eigenvector of $U$ with eigenvalue $1$.

To use quantum walk for search, we extend the step of the algorithm with a query to an oracle, making the step
$$
U' = U \cdot (Q \otimes I_4) .
$$
Here $Q$ is the query transformation which flips the sign at a marked vertex, irrespective of the coin state. 
Note that $\ket{\psi_0}$ is a 1-eigenvector of $U$ but not of $U'$.
If there are marked vertices, the state of the algorithm starts to deviate from $\ket{\psi(0)}$.
In case of a single marked vertex, after $O(\sqrt{N\log{N}})$ steps the inner product $\braket{\psi(t)}{\psi(0)}$ becomes close to $0$.
If one measures the state at this moment, he will find the marked vertex with $O(1 / \log{N})$ probability~\cite{Ambainis:2005}.
With amplitude amplification this gives the total running time of $O(\sqrt{N} \log{N})$ steps.


\subsection{Lackadaisical quantum walk}

In case of lackadaisical quantum walk the coin subspace of the walk is 5-dimensional Hilbert space spanned by the set of states $\{\ket{c}: c\in \{\U,\D,\L,\R,\S \}\}$. The  Hilbert space of the quantum walk is $\mathbb{C}^N\otimes \mathbb{C}^5$.

The shift operator acts on a self loop as 
\begin{equation}
S\ket{i,j,\S} = \ket{i,j,\S} .
\end{equation}
The coin operator is 
\begin{equation}
C = 2 \ket{s_c}\bra{s_c} - I_5
\end{equation}
with 
$$
\ket{s_c} = \frac{1}{\sqrt{4 + l}}(\ket{\U} + \ket{\D} + \ket{\L} + \ket{\R} + \sqrt{l}\ket{\S}) .
$$
The system starts in 
\begin{equation}\label{eq:psi0_grid}
\ket{\psi(0)} = \frac{1}{\sqrt{N}} \sum_{i,j=0}^{\sqrt{N}-1} \ket{i,j} \otimes \ket{s_c} ,
\end{equation}
which is uniform distribution over vertices, but not directions. As before $\ket{\psi(0)}$ is a unique 1-eigenvector of $U$.

The step of the search algorithm is $U' = U \cdot (Q \otimes I_5)$.
As it is shown in \cite{Wong:2018}, in case of a single marked vertex, for the weight $l = \frac{4}{N}$, after $O(\sqrt{N\log{N}})$ steps the inner product $\braket{\psi(t)}{\psi(0)}$ becomes close to $0$.
If one measures the state at this moment, he will find the marked vertex with $O(1)$ probability, which gives $O(\log{N})$ improvement over the loopless algorithm.


\section{Stationary states of the lackadaisical quantum walk}

In this section we will show that the lackadaisical quantum walk, similarly to the regular (non-lackadaisical) quantum walk, has exceptional configurations, i.e. placements of marked vertices for which the walk have no speed-up over the classical exhaustive search.

Consider a state $\ket{s_{\U}^a} = a(- (3+l)\ket{\U} + \ket{\D} + \ket{\L} + \ket{\R} + \sqrt{l}\ket{\S})^T$.
Similarly one can define states $\ket{s_{\D}^a}$, $\ket{s_{\L}^a}$ and $\ket{s_{\R}^a}$. The defined states are orthogonal to $\ket{s_c}$.
Consider an effect of the coin transformation on $\ket{s_{\U}^a}$:
$$
C \ket{s_{\U}^a} = \left( 2 \ket{s_c}\bra{s_c} - I_5 \right) \ket{s_{\U}^a} = - \ket{s_{\U}^a} .
$$
As one can see, the coin transformation inverts a sign of the state.

Now, consider a two-dimensional grid with two marked vertices $(i,j)$ and $(i+1,j)$.
Let $\ket{\phi_{stat}^a}$ be a state where the coin part of all unmarked vertices is $\ket{s_c^a} = a(\ket{\U} + \ket{\D} + \ket{\L} + \ket{\R} + \sqrt{l}\ket{\S})^T$, the coin part of $(i,j)$ is $\ket{s_{\R}^{a}}$ and the coin part of $(i+1,j)$ is $\ket{s_{\L}^{a}}$ (see Fig.~\ref{fig:Two_marked_vertices}), that is,

\begin{figure}[!htb]
\centering
\includegraphics[scale=0.25]{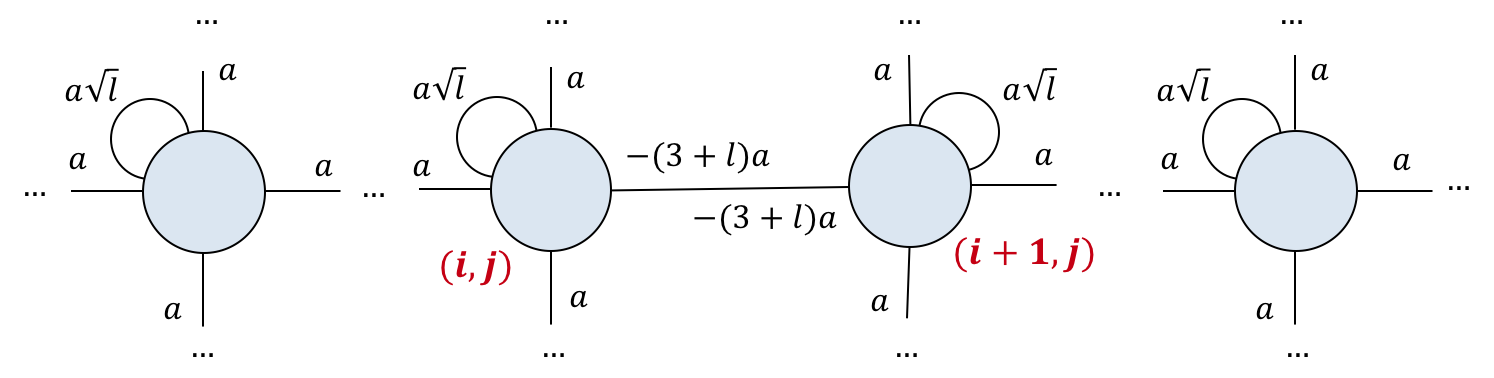}
\caption{Stationary state of two marked vertices $(i,j)$ and $(i+1,j)$.}
\label{fig:Two_marked_vertices}
\end{figure}

\begin{equation}\label{eq:phi_grid}
\ket{\phi_{stat}^a} = \sum_{i,j=0}^{\sqrt{N}-1} \ket{i,j} \ket{s_c^a} - (4 + l)a \left( \ket{i,j,\R} + \ket{i+1,j,\L} \right).
\end{equation}
We claim that this state is not changed by a step of the algorithm.

\begin{lemma}\label{lemma:grid}
Consider a grid of size $\sqrt{N} \times \sqrt{N}$ with two adjacent marked vertices $(i,j)$ and $(i+1,j)$. Then the state $\ket{\phi_{stat}^a}$, given by Eq.~(\ref{eq:phi_grid}), is not changed by the step of the algorithm, that is, $U'\ket{\phi_{stat}^a} = \ket{\phi_{stat}^a}$.
\end{lemma}

\begin{proof}
Consider the effect of a step of the algorithm on $\ket{\phi_{stat}^a}$. The query transformation flips the sign of marked vertices. The coin transformation has no effect on $\ket{s_c^a}$ but flips the signs of $\ket{s_{\L}^{a}}$ and $\ket{s_{\R}^{a}}$. Thus, $(I\otimes C)(Q\otimes I)$ does not change the amplitudes of unmarked vertices and twice flips the signs of amplitudes of marked vertices. 
Therefore, we have $(I\otimes C)(Q\otimes I)\ket{\phi_{stat}^a} = \ket{\phi_{stat}^a}.$
The shift transformation swaps the amplitudes of near-by vertices. For $\ket{\phi_{stat}^a}$, it swaps $a$ with $a$ and $-(3+l)a$ with $-(3+l)a$. Thus, we have $S(I\otimes C)(Q\otimes I)\ket{\phi_{stat}^a} = \ket{\phi_{stat}^a}$.
\end{proof}

The initial state of the algorithm, given by Eq.~(\ref{eq:psi0_grid}),
can be written as 
\begin{equation}
\ket{\psi_0} = \ket{\phi_{stat}^a} + (4+l)a(\ket{i,j,\R} + \ket{i+1,j,\L}),
\end{equation}
for $a=1/\sqrt{(4+l)N}$. The only part of the initial state which is changed by the step of the algorithm is
\begin{equation}
\frac{\sqrt{4+l}}{\sqrt{N}}(\ket{i,j,\R} + \ket{i+1,j,\L}).
\end{equation}

\noindent
Let us establish an upper bound on the probability of finding a marked vertex. 

\begin{lemma}
Consider a grid of size $\sqrt{N} \times \sqrt{N}$ with two adjacent marked vertices $(i,j)$ and $(i+1,j)$. Then for any number of steps, the probability of finding a marked vertex $p_M$ is $O\left(\frac{1}{N}\right)$.
\end{lemma}

\comment{
  Alternatively we can skip the proof and refer to the result/argument from \cite{Nahimovs:2017}
}

\begin{proof}
We have $M=\{(i,j),(i+1,j)\}$.
The only part of the initial state $\ket{\psi(0)}$ changed by the step of the algorithm is $\ket{\phi} = (4+l)a(\ket{i,j,\R} + \ket{i+1,j,\L})$.
The basis states $\ket{i,j,\R}$ and $\ket{i+1,j,\L}$ have the biggest amplitudes of $-(3+l)a$ in the stationary state. 
Therefore, the maximum probability of finding a marked vertex is reached if the state $\ket{\phi}$ becomes
\begin{equation}
\ket{\phi'} = -\alpha\ket{i,j,\R}-\beta\ket{i+1,j,\L},
\end{equation}
for $\alpha,\beta \geq 0$.
Thus, $p_M$ is at most
\begin{equation}
\label{eq:p_M_grid}
p_M\leq 6a^2 + 2(a\sqrt{l})^2 + \left(-(3+l)a-\alpha\right)^2 + \left(-(3+l)a-\beta\right)^2.
\end{equation}

\noindent
Since the evolution is unitary, we have $\alpha^2+\beta^2 = || \ket{\phi} ||^2 = 2\left( (4+l)a \right)^2 $.
Due to symmetry $\alpha$ and $\beta$ should be equal, so the expression (\ref{eq:p_M_grid}) reaches the maximum when $\alpha = \beta = \sqrt{2}(4+l)a$. 

We have $l = \frac{4}{N}$ and $a = \frac{1}{\sqrt{(4+l)N}} = \frac{1}{\sqrt{4(N+l)}}$. Each of summands in the expression (\ref{eq:p_M_grid}) is $O\left(\frac{1}{N}\right)$ and, therefore, we have $p_M = O\left(\frac{1}{N}\right)$.
\end{proof}

\noindent
That is the probability of finding a marked vertex is of the same order as for the classical exhaustive search.

Note that if we have a block of marked vertices we can construct a stationary state as long as we can tile the block by the sub-blocks of size $1 \times 2$ and $2 \times 1$. For example, consider $M = \{(0,0),(1,0),(1,1),(1,2)\}$ for $n \geq 3$. Then the stationary state is given by
$$
\ket{\phi_{stat}^a} = \sum_{i,j=0}^{n-1} \ket{i,j} \ket{s_c^a} - (4 + l)a \left( \ket{0,0,\R} + \ket{1,0,\L} + \ket{1,1,\U} + \ket{1,2,\D} \right).
$$
For more details on constructions of stationary states for blocks of marked vertices on two-dimensional grid see~\cite{Nahimovs:2017}. The paper focuses on the non-lackadaisical quantum walk, nevertheless, the results can be  easily extended to the lackadaisical quantum walk.


\section{Optimality of $l$ for multiple marked vertices}

In \cite{Wong:2018} Wong showed that in case of a single marked vertex, for the weight $l = \frac{4}{N}$, after $O(\sqrt{N\log{N}})$ steps the inner product $\braket{\psi(t)}{\psi(0)}$ becomes close to $0$.
If one measures the state at this moment, he will find the marked vertex with $O(1)$ probability\footnote{The numerical results in \cite{Wong:2018} show that probability of finding a marked vertex is close to $1$ and approaches $1$ then $N$ goes to infinity.}. 
The suggested value of $l$, however, is optimal for a single marked vertex only. 
Saha~{\it et al.}\cite{Saha:2018} studied search for a block of $\sqrt{m} \times \sqrt{m}$ marked vertices and showed that optimal weight in this setting is $l = \frac{4}{N(m + \sqrt{m}/2)}$.

In this section we study search for an arbitrary placement of multiple marked vertices. The presented data is obtained from numerical simulations. The values listed in the tables are calculated in the following way. The number of steps of the algorithm $T$ is the
smallest $t$ for which $\braket{\psi(t)}{\psi(0)}$ reaches 0 (becomes negative). By the probability
we mean the probability of finding a marked vertex when $\ket{\psi(T)}$ is measured.

Tables \ref{tab:2_random_marked_vertices} and \ref{tab:3_random_marked_vertices} give the number of steps and the probability of finding a marked vertex for random placements of $2$ and $3$ marked vertices on $100 \times 100$ grid for $l = \frac{4}{N}$. As one can see the probability of finding a marked vertex is no more close to $1$ as it is for a single marked vertex.

\begin{table}[h]
\centering
\begin{tabular}{lll}
Marked vertices  & $T$ & $Pr$  \\  
(0, 0), (23, 27) & 153 & 0.586377681077719 \\
(0, 0), (35, 68) & 150 & 0.591030741055657 \\
(0, 0), (30, 69) & 151 & 0.588384716869901 \\
(0, 0), (42, 4)  & 152 & 0.590451037529614 \\
(0, 0), (84, 60) & 151 & 0.584982804352049
\end{tabular}
\caption{The number of steps and the probability of finding a marked vertex for different placements of two marked vertices for $100 \times 100$ grid for $l = \frac{4}{N}$.}
\label{tab:2_random_marked_vertices}
\end{table}

\begin{table}[h]
\centering
\begin{tabular}{lll}
Marked vertices  & $T$ & $Pr$  \\  
(0, 0), (34, 52), (93, 53) & 117 & 0.440756928151790 \\
(0, 0), (26, 12), (22, 32) & 126 & 0.434581723157292 \\
(0, 0), (40, 94), (13, 62) & 119 & 0.430688837061525 \\
(0, 0), (7, 44),  (7, 98)  & 131 & 0.430029225026132 \\
(0, 0), (80, 78), (28, 31) & 118 & 0.454915029501263 
\end{tabular}
\caption{The number of steps and the probability of finding a marked vertex for different placements of three marked vertices for $100 \times 100$ grid for $l = \frac{4}{N}$.}
\label{tab:3_random_marked_vertices}
\end{table}

\noindent
Table \ref{tab:1_10_marked_vertices_Wong_Saha} shows the number of steps and the probability for $200 \times 200$ grid with the set of marked vertices 
\begin{equation}
M_m = \{ (0,10i) \  | \ i \in [0, m-1] \}
\end{equation}
for weights of a self-loop suggested by Wong (the 2nd and the 3rd columns) and by Saha~{\it et al.} (the last two columns). 
As one can see for both weights the probability goes down with the number of marked vertices. 

\begin{table}[h]
\centering
\begin{tabular}{lllll}
 & \multicolumn{2}{l}{$l = \frac{4}{N}$} & \multicolumn{2}{l}{$l = \frac{4}{N(m+\sqrt{m}/2)}$} \\  
$m$ & $T$ & $Pr$ & $T$ & $Pr$  \\  
 1 & 602 & 0.987103466750771 & 602 & 0.9871034667507710 \\
 2 & 374 & 0.556471227830710 & 355 & 0.3290596740364150 \\
 3 & 320 & 0.393873564782729 & 307 & 0.1901285270921410 \\
 4 & 288 & 0.318205769345174 & 278 & 0.1362737798676850 \\
 5 & 266 & 0.269653054659757 & 258 & 0.1120867687513450 \\
 6 & 250 & 0.234725633256426 & 243 & 0.0963898188447711 \\
 7 & 235 & 0.205158185237765 & 229 & 0.0847091122232096 \\
 8 & 223 & 0.184324335272977 & 218 & 0.0764074018340319 \\
 9 & 213 & 0.168420810292804 & 208 & 0.0694735116911546 \\
10 & 203 & 0.153267792359668 & 198 & 0.0634301283171891 
\end{tabular}
\caption{The number of steps and the probability of finding a marked vertex for $200 \times 200$ grid with the set of marked vertices $M_m$ for different $l$.}
\label{tab:1_10_marked_vertices_Wong_Saha}
\end{table}

We tried to adjust the value of $l$ to increase the probability of finding a marked vertex. 
We searched for a better value of $l$ in the form $l = \frac{4}{N}a$. The figures \ref{fig:m_2_different_a} and \ref{fig:m_3_different_a} show the probability of finding a marked vertex for $100 \times 100$ grid with the sets of marked vertices $M_2$ and $M_3$, respectively, for different values of $a$.

\begin{figure}[!htb]
\centering
\includegraphics[scale=0.5]{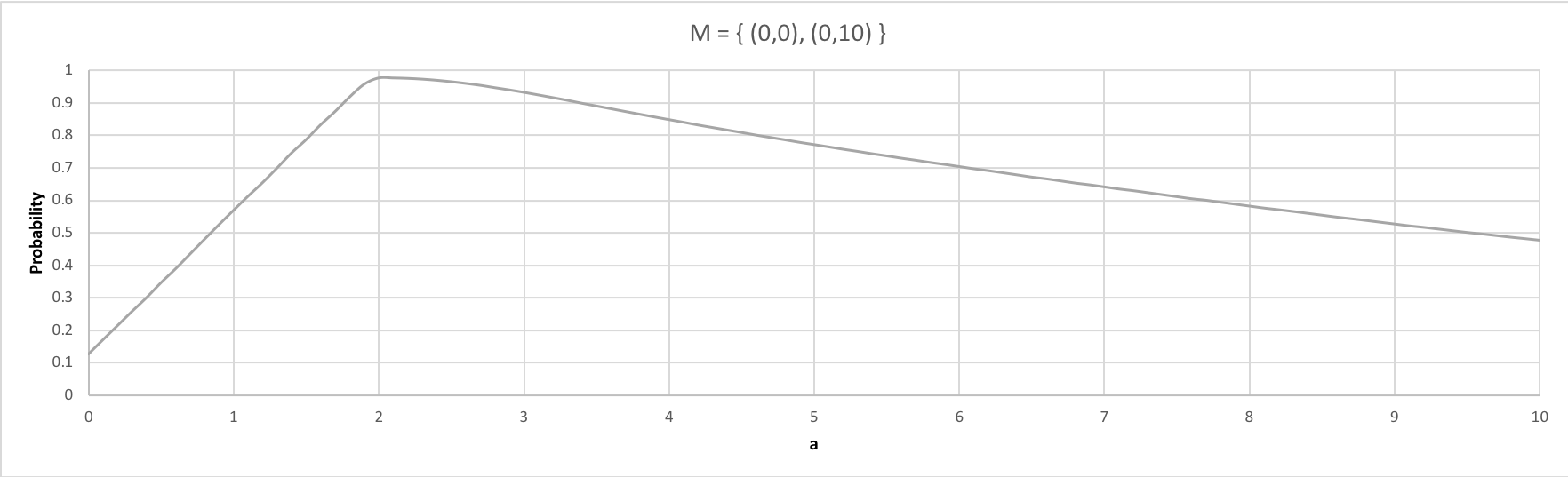}
\caption{Probability of finding a marked vertex for $100 \times 100$ grid with the set of marked vertices $M_2$ for different values of $a$.}
\label{fig:m_2_different_a}
\end{figure}

\begin{figure}[!htb]
\centering
\includegraphics[scale=0.5]{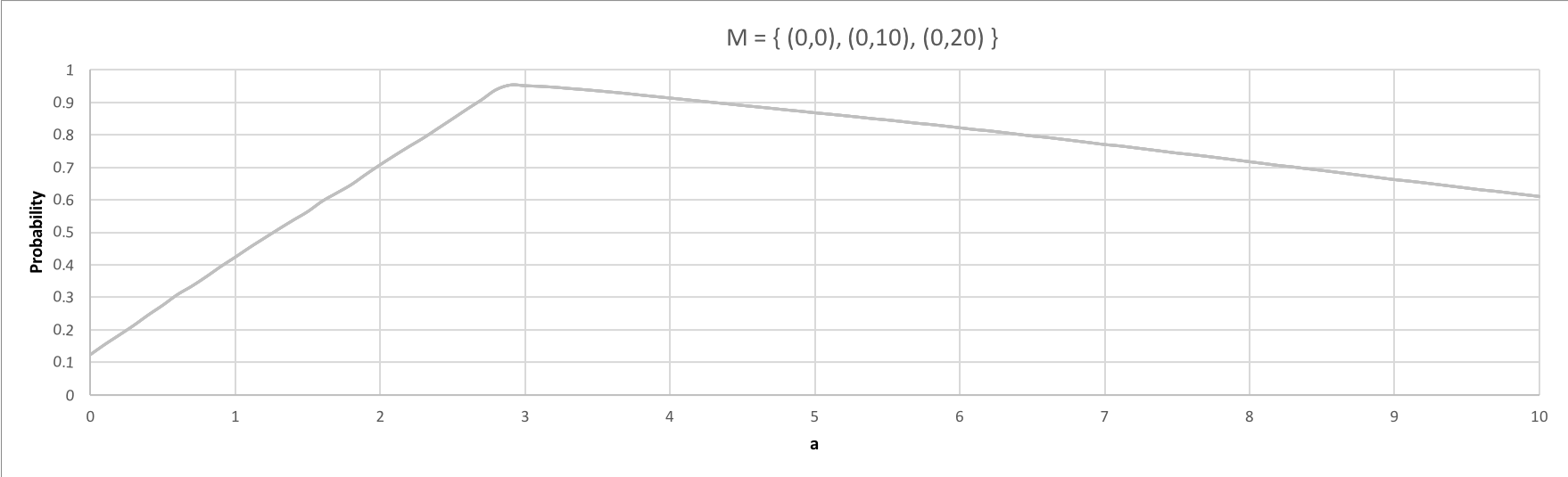}
\caption{Probability of finding a marked vertex for $100 \times 100$ grid with the set of marked vertices $M_3$ for different values of $a$. }
\label{fig:m_3_different_a}
\end{figure}

\noindent
As one can see the optimal value of $a$ for $M_2$ is close to $2$ and for $M_3$ is close to $3$.
The similar results were obtained for bigger grids with larger sets of marked vertices. 
Table \ref{tab:2_10_marked_vertices_optimal_a} gives the optimal value of $a$ and the corresponding number of steps and the probability for $200 \times 200$ grid with the set of marked vertices $M_m$.

\begin{table}[h]
\centering
\begin{tabular}{llll}
$m$ & $a_{opt}$ & $T$ & $Pr$  \\  
 2 & 1.94 & 470 & 0.970784853767743 \\
 3 & 2.90 & 419 & 0.968156591210997 \\
 4 & 3.82 & 394 & 0.957428109231279 \\
 5 & 4.66 & 374 & 0.93524432034913; \\
 6 & 5.44 & 358 & 0.910278544128265 \\
 7 & 6.17 & 329 & 0.884824083920976 \\
 8 & 7.06 & 301 & 0.884650346189075 \\
 9 & 8.00 & 295 & 0.891195819702051 \\
10 & 8.86 & 292 & 0.889060897077511
\end{tabular}
\caption{The number of steps and the probability of finding a marked vertex for $200 \times 200$ grid with the set of marked vertices $M_m$ for optimal $a$.}
\label{tab:2_10_marked_vertices_optimal_a}
\end{table}

This raises a conjecture that the optimal weight of a self loop is $l = \frac{4(m - O(m))}{N}$. The table \ref{tab:1_10_marked_vertices_NN} shows the number of steps and the probability for $200 \times 200$ grid with the set of marked vertices $M_m$ for the weight $l=\frac{4m}{N}$ (the 2nd and the 3rd columns) and $l=\frac{4m-\sqrt{m}}{N}$ (the last two columns).
As one can see $l = \frac{4m}{N}$ results in a high probability of finding a marked vertices for a small number of marked vertices, however, the probability goes down with the number of marked vertices. 
On the other hand, $l = \frac{4(m-\sqrt{m})}{N}$ gives a modest probability for a small number of marked vertices, but the probability grows with the number of marked vertices (and, moreover, seems to tend to a constant). Therefore, we would suggest to use the last of the proposed value of $l$, especially for bigger grids and large number of marked vertices.

It worth noting, that the found values of $l$ result in high probability not only for $M_i$ sets of marked vertices, but work equivalently well for other placements of $m$ marked vertices, including a random placement.

\begin{table}[h]
\centering
\begin{tabular}{lllll}
 & \multicolumn{2}{l}{$l = \frac{4m}{N}$} & \multicolumn{2}{l}{$l = \frac{4(m-\sqrt{m})}{N}$} \\  
$m$ & $T$ & $Pr$ & $T$ & $Pr$  \\  
 1 & 602 & 0.987103466750771 & 421 & 0.138489015636136 \\
 2 & 480 & 0.973610115577208 & 358 & 0.368553562270952 \\
 3 & 426 & 0.970897595293325 & 326 & 0.474753065755793 \\
 4 & 400 & 0.957956584718826 & 305 & 0.541044821578945 \\
 5 & 376 & 0.933005243569973 & 288 & 0.593276362658860 \\
 6 & 352 & 0.904811189309431 & 277 & 0.633876384394702 \\
 7 & 312 & 0.885901799105365 & 268 & 0.661120674334215 \\
 8 & 300 & 0.891698403206386 & 260 & 0.678417412900138 \\
 9 & 296 & 0.892165251874117 & 254 & 0.694145864271432 \\
10 & 293 & 0.884599315314024 & 250 & 0.709033853082403 
\end{tabular}
\caption{The number of steps and the probability of finding a marked vertex for $200 \times 200$ grid with the set of marked vertices $M_m$ for different $l$.}
\label{tab:1_10_marked_vertices_NN}
\end{table}

\comment {
The complete data is available in this paper’s arXiv source.
The code used to simulate the quantum walk is available  on Github.
}


\section{Conclusions}\label{sec:conclusions}

In this paper, we have demonstrated the existence of exceptional configurations of marked vertices for search by lackadaisical quantum walk on two-dimensional grid. 
We also numerically showed that weights of the self-loop $l$, suggested by previous papers~\cite{Wong:2018,Saha:2018}, are not optimal for multiple marked vertices (both weight seems to work in specific cases only). We proposed two values of $l$ resulting in a much higher probability of finding a marked vertex than previously suggested weights. Moreover, for the found values, the probability of finding a marked vertex does not decrease with number of marked vertices.


\comment {
\subparagraph*{Acknowledgements.}
}


\bibliography{Paper}


\end{document}